\documentclass[11pt,reqno,a4paper]{article}

\usepackage{amsmath}
\usepackage{xcolor}
\usepackage{amssymb}
\usepackage{graphicx}
\usepackage{verbatim}
\usepackage[T1]{fontenc}
\usepackage[utf8]{inputenc}

\newcommand{\eps}{\varepsilon}

\renewcommand{\phi}{\varphi}
\renewcommand{\bar}[1]{\overline{#1}}
\renewcommand{\leq}{\leqslant}
\renewcommand{\geq}{\geqslant}

\newcommand{\Prob}{\mathbb{P}}
\newcommand{\Exp}{\mathbb{E}}

\newcommand{\LebSph}{\lambda}
\newcommand{\Sph}{\mathcal{S}}
\newcommand{\Ball}{\mathcal{B}}

\newcommand{\B}{\mathrm{B}}

\newcommand{\LebNum}{\mathcal{N}}

\newcommand{\nr}[1]{\left\Vert #1\right\Vert}
\newcommand{\abs}[1]{\left\vert #1\right\vert}

\renewcommand{\d}{\mathrm{d}}

\newcommand{\dH}{\d_{\mathrm{H}}}

\newcommand{\Rsp}{\mathbb{R}}

\DeclareMathOperator{\BigO}{O}

\newcommand{\dd}{\mathrm{d}}
\newcommand{\hh}{\mathrm{h}}
\DeclareMathOperator{\conv}{conv}

\usepackage{xspace}

\newcommand{\const}{\mathrm{const}}

\usepackage{bibentry}

\newcommand{\ignore}[1]{}
\newcommand{\nobibentry}[1]{{\let\nocite\ignore\bibentry{#1}}}

\usepackage{amsthm}
\newtheorem*{theoremA}{\sc{Theorem} \ref{th:main}}

\newtheorem{theorem}{\sc{Theorem}}[section]
\newtheorem{lemma}[theorem]{\sc{Lemma}}

\newtheorem{corollary}[theorem]{\sc{Corollary}}

\theoremstyle{definition}
\newtheorem{definition}{Definition}[section]

\newcommand{\Proof}{Proof}
\newcommand{\tco}[1]{}
\newcommand{\oco}[1]{#1}

\title{Lower Bounds for $k$-Distance Approximation}
\author{Quentin  M\'erigot}

\newcommand{\copyrighttext}{}

\begin{document}

\copyrighttext

\maketitle

\begin{abstract}
Consider a set $P$ of $N$ random points on the unit sphere of
dimension $d-1$, and the symmetrized set $S = P\cup(-P)$. The
\emph{halving polyhedron} of $S$ is defined as the convex hull of the
set of centroids of $N$ distinct points in $S$. We prove that after
appropriate rescaling this halving polyhedron is Hausdorff close to
the unit ball with high probability, as soon as the number of points
grows like $\Omega(d\log(d))$.  From this result, we deduce
probabilistic lower bounds on the complexity of approximations of the
distance to the empirical measure on the point set by distance-like functions.
\end{abstract}

\tco
{
\category{I.3.5}{Compu\-ter Graphics}{Computational Geometry and Object Modeling}
\keywords{$k$-distance, order-$k$ Voronoi diagram, Chernoff bound}
}

\section{Introduction}

The notion of \emph{distance to a measure} was introduced in order to
extend existing geometric and topological inference results from the
usual Hausdorff sampling condition to a more probabilistic model of
noise \cite{distance-to-measure}. Consider a finite subset $P$ of the
Euclidean space $\Rsp^d$ and a positive number $k$ in the range
$\{1,\hdots,\abs{P}\}$, where $\abs{P}$ denotes the cardinality of~$P$. The distance to the empirical measure on $P$ is given by the
following formula:
\begin{equation}
\dd_{P,k}(x) := \left(\frac{1}{k} \left[ \min_{p_1,\hdots, p_k \in P}
  \sum_{i=1}^k \nr{x - p_i}^2 \right]\right)^{1/2},\label{eq:dPk}
\end{equation}
where the minimum is taken over the sets consisting of $k$ distinct
points in $P$. We will call this function the \emph{$k$-distance to
  the point set $P$}. Equation~\eqref{eq:dPk} allows to compute the
value of the $k$-distance at a certain point $x$ easily, using a
nearest-neighbor data structure. On the other hand, this formula
cannot be used to perform more global computations, such as estimating
the Betti numbers of a sublevel set $\dd_{P,k}^{-1}(0,r) = \{x \in
\Rsp^d;~ \dd_{P,k}(x) \leq r\}$, not to mention reconstructing a
simplicial complex homotopic to this set.

There is another representation of the $k$-distance that allows to
perform such global operations. It is computational geometry folklore
that $\dd_{P,k}$ can be rewritten as the square root the minimum of a
finite number of quadratic functions. More precisely $\dd_{P,k}(x)^2 =
\min_{\bar{p}} \nr{x - \bar{p}}^2 + w_{\bar{p}}$, where the minimum is
taken over the set of centroids of $k$ distinct points in $P$, and
$w_{\bar{p}}$ is chosen adequately (see~\S\ref{sec:kdist}). This
implies that sublevel sets of $\dd_{P,k}$ are simply union of balls,
and this allows one to compute their homotopy type using weighted
alpha-complexes or similar constructions
\cite{edelsbrunner1992alpha}. However, since the number of centroids
of $k$ distinct points in $P$ grows exponentially with the number of
points, this formulation is not very practical either.

Fortunately, many geometric and topological inference result continue
to hold if one replaces $\dd_{P,k}$ in the computation by a ``good
approximation'' $\phi$~\cite{distance-to-measure}. This means that the
error $\nr{\phi - \dd_{P,k}}_\infty := \max_{\Rsp^d} \abs{\phi -
  \dd_{P,k}}$ has to be small enough, and that the approximating
function $\phi$ should be \emph{distance-like}. For the purpose of
this work, $\phi$ is distance-like if there exists a finite set of
\emph{sites} $Q\subseteq \Rsp^d$ and non-negative \emph{weights}
$(w_q)_{q\in Q}$ such that $\phi = \phi_{Q}^w$, where $\phi_Q^w$ is
defined by
\begin{equation}
\phi_{Q}^w(x) := \left(\min_{q \in Q} \nr{x-q}^2 +
w_q\right)^{1/2},~~w_q\geq 0. \label{eq:dl}
\end{equation}
To summarize, in order to estimate the topology of the sublevel set of
the $k$-distance, it makes sense to try to replace it by a
distance-like function that uses much fewer sites.

Note that one could try to find approximations of the $k$-distance in
a class of functions $\mathcal{F}$ different from the class of
distance-like functions. It is indeed possible that a well chosen
class of functions would produce more compact approximations of the
$k$-distance and of similar functions. However, changing the class of
function would practically forbid to use these approximations for the
purpose of geometric inference, because of the lack of (i)
computational topology tools to compute with the sublevel sets of
functions in $\mathcal{F}$ and (ii) a geometric inference theory
adapted to this class of function.

\tco{\smallskip}
\paragraph{Complexity of $k$-distance} The natural formalization of our approximation problem
is as follows.  Given a finite point set $P$ in $\Rsp^d$, a number
$k>0$ and a target approximation error $\eps$, what is the minimum
cardinality of a weighted point set $(Q,w)$ with non-negative weights
such that the approximation error $\nr{\phi_{Q}^w - \dd_{P,k}}_\infty$
is bounded by $\eps$ ? We call this cardinality the $\eps$-complexity
of the $k$-distance function $\dd_{P,k}$.  When $\eps$ is zero, the
$0$-complexity of the $k$-distance function $\dd_{P,k}$ is equal to
the number of order-$k$ Voronoi cell of $P$ that have non-empty
interior. One can then translate lower-bounds on the number of
order-$k$ Voronoi cells into lower bounds for the $0$-complexity of
the $k$-distance.
The purpose of this article is to provide lower bounds on
$\eps$-complexity of the $k$-distance function for a \emph{non-zero
  approximation error $\eps$}, when $P$ is a random point cloud on the
unit $(d-1)$-dimensional sphere and $k = \abs{P}/2$.

\subsection{Prior work}

\paragraph{Approximation of the $k$-distance}
The question of approximating the distance to the measure by a
distance-like function with few sites has been originally raised in
\cite{witnessed}. In this article, the authors proposed an
approximation of the $k$-distance, called the \emph{witnessed
  $k$-distance} and denoted by $\dd_{P,k}^{\mathrm{w}}$, which
involves only a \emph{linear} number of sites. They also give a
probabilistic upper bound on the approximation error
$\Vert\dd_{P,k}^{\mathrm{w}}-\dd_{P,k}\Vert_\infty$ under the
hypothesis that the point cloud $P$ is obtained by sampling a
$\ell$-dimensional submanifold of the Euclidean space. The upper bound
on the approximation error degrade as the intrinsic dimension $\ell$
of the underlying submanifold increases. This 
suggests that
approximating the distance to the uniform measure on a point cloud
drawn from a high-dimensional submanifold might be difficult.

It is possible to build data structures that allow to compute
approximate pointwise values of the $k$-distance in time that is
logarithmic in the number of points --- but exponential in the ambient
dimension \cite{hp12}. The same data structure can be used to compute
generalizations of the $k$-distance, such as the sum of the $p$th
power to the $k$-nearest neighbors for an exponent $p$ larger than
one.

\tco{\smallskip}
\paragraph{Complexity of order-$k$ Voronoi}
As mentioned earlier, there exists upper and lower bounds for the
number of cells in an order-$k$ Voronoi diagrams, and those bounds can
be translated into bounds on the number of sites that one needs to use
in order to get an exact representation of the $k$-distance function
by a distance-like function. When $k$ is half the cardinality of the
point cloud $P$, we will speak of \emph{halving Voronoi diagram}. A
\emph{halving hyperplane} for $P$ is a hyperplane that separates $P$
into two sets with equal cardinality. The number of halving
hyperplanes yielding different partitions of $P$ is a lower bound on
the number of infinite halving Voronoi cells.  The best lower bound on
the number of halving hyperplanes, that holds for an arbitrary ambient
dimension $d$, is given by $N^{d-1} \mathrm{e}^{\Omega(\sqrt{\log
    N})}$, where $N$ is the cardinality in the point set
\cite{toth2001point}. This bound improves on previous lower bounds by
several authors, e.g.  \cite{erdos1973dissection,
  edelsbrunner1985number,seidel1987personal}.

In \cite{barany1994expected}, the authors study the expectation of the
number of $k$-sets of a point cloud $P$ obtained by sampling
independent random point on the sphere. Recall that a $k$-set is a
subset of $k$ points in $P$ that can be separated from other points in
$P$ by a hyperplane; in particular, each $k$-set corresponds to an
infinite order-$k$ Voronoi cell. The authors prove that the expected
number of $k$-sets in a random point cloud $P$ on the sphere is upper
bounded by $\BigO(\abs{P}^{d-1})$.

In order to obtain our lower bounds on the number of sites needed to
approximate $\dd_{P,k}$, we will use the notion of $k$-set polyhedron,
originally introduced in \cite{edelsbrunner1997cutting}. This
polyhedron is the convex hull of the set of centroids of $k$ distinct
points in $P$. The relation between this polyhedron and the notion of
$k$-set is that the number of extreme points of the $k$-set polyhedron
is equal to the number of $k$-sets in the point set. When $k$ is half
the cardinality of $P$, we will call this polyhedron the halving
polyhedron.

\subsection{Contributions}
The main result of this work concerns the geometry of halving
polyhedra of finite point sets on the unit sphere. Our theorem shows
that even with relatively few points, the halving polyhedron of a
certain random point set $S$ on the unit $(d-1)$-dimensional sphere is
Hausdorff-close to a ball with high probability. The random point set
$S$ is obtained by picking $N$ random, independent and uniformly
distributed points $p_1,\hdots,p_N$ on the $(d-1)$-dimensional unit
sphere, and by letting $S = \{ \pm p_i\}_{1\leq i\leq N}$. The halving
polyhedron of $S$ is a by definition a random convex polyhedron, which
we denote by $L_N^d$.  The statement and proof of this theorem are
inspired by the main theorem of \cite{artstein2006geometric}. We use
the quantity
\begin{equation*}
m_{d} := \Exp(\abs{X \cdot u}) \simeq \left(\frac{2}{\pi d}\right)^{1/2},
\end{equation*}
where $X$ is a uniformly distributed random vector on the unit
$(d-1)$-sphere, and $u$ is an arbitrary unit vector. Note that this
quantity turns out not to depend on $u$.

\smallskip

\begin{theoremA}
There exists an absolute constants $c> 0$ such that for every
positive number $\eta$, the inequality
$$\dH\left(\frac{1}{m_{d}} L_{N}^d, \Ball(0,1)\right)
  \leq \eta$$
holds with probability at least $1- 2 \exp\left[c \cdot \left(d
    \log(1/\delta) - N \eta^2\right)\right],$
where $\delta = \min(\eta,1/\sqrt{d})$.
\end{theoremA}

In Section~\ref{sec:dtm}, we deduce from this theorem a probabilistic
lower bound on the $\eps$-complexity of the $\dd_{S,N}$, where $S$ is
the point set defined in the previous paragraph. The exact statement
of the lower bound can be found in Theorem~\ref{th:approx}.



\section{Traces at infinity and $k$-set \tco{\\} polyhedra}
\paragraph{Background: support function}
The support function of a convex subset $K$ of $\Rsp^d$ is a function
$x\mapsto \hh(K,x)$ from the unit sphere to $\Rsp$. It is defined by
the following formula: 
\begin{equation}\hh(K,u) := \max \left\{ x \cdot u;~ x \in
K\right\}.
\end{equation}
The application that maps a convex set to its support
function on the sphere satisfies the following isometry property:
\begin{equation}
\nr{\hh(K,.) - \hh(L,.)}_{\infty,\Sph^{d-1}} = \dH(K,L).
\label{eq:isom}
\end{equation}
where $\nr{f}_{\infty,\Sph^{d-1}} = \max_{\Sph^{d-1}} \abs{f}$ is the infinity
norm on the unit sphere, and where $\dH$ denotes the
\emph{Hausdorff distance}. The proof of this equality is given in
Theorem~1.8.11 in \cite{schneider1993convex}.  All the elementary
facts about support functions that we will need can be found in the
first chapter of this book.

\subsection{\tco{\hspace{-.2cm}}Traces at infinity of distance-like functions.}
We call \emph{distance-like} a non-negative function whose square can
be written as the minimum of a family of unit paraboloids
$\nr{x-q}^2+w_q$, with $w_q\geq 0$. Note that this definition is
equivalent to the one given in \cite{distance-to-measure}, thanks to
the remark following Proposition~3.1 in this article.

Given a finite subset $Q$ of the Euclidean space and non-negative
weights~$w$, we let $\phi := \phi_Q^w$ be the distance-like function
defined by \eqref{eq:dl}. We call \emph{trace at infinity} of $\phi$
and denote by $K(\phi)$ the convex polyhedra obtained by taking the
convex hull of the set of sites $Q$ used to define $\phi$. 

The name trace at infinity is explained by the following asymptotic
development of the values of $\phi$, along a unit-speed ray starting
at the origin, i.e. $\gamma_t := tu$ with $\nr{u}=1$:
\begin{align}
\phi(\gamma_t) &= \min_{q\in Q} \left(\nr{tu - q}^2 + w_q\right)^{1/2} \notag \\
&= \min_{q\in Q} t\left(1 - \frac{2}{t} u \cdot q + \frac{w_q+\nr{q}^2}{t^2}\right)^{1/2}\notag\\
&= t - \max_{q\in Q} u \cdot q + \mathrm{O}(1/t)\\
&= t - \hh(K(\phi),u) + \mathrm{O}(1/t)
\label{eq:asymp}
\end{align}
The last equality follows from the fact that the support function of
the convex hull of a set $Q$ is given by $\max_{q\in Q} u \cdot q$.
Using these computations, we get the following lemma. This lemma
implies in particular that if two functions $\phi:=\phi_{Q}^w$ and
$\psi:=\phi_{P}^w$ coincide on $\Rsp^d$, then $K(\phi) = K(\psi)$.

\begin{lemma}
Given two distance-like functions $\phi := \phi_{Q}^w$ and $\psi :=
\phi_{P}^v$ as in \eqref{eq:dl},
\begin{equation}
\dH(K(\phi), K(\psi)) \leq \nr{\phi - \psi}_\infty.
\end{equation}
\tco{\vspace{-.6cm}}
\label{lemma:laguerre}
\end{lemma}
\begin{proof} The asymptotic developments for $\phi$ and $\psi$ given in
equation~\eqref{eq:asymp} imply the following inequality for every unit
direction $u$: $ \abs{\hh(K(\phi), u) - \hh(K(\psi), u)} \leq \nr{\phi
  - \psi}_{\infty}$. With Equation~\eqref{eq:isom}, this yields the
desired bound on the Hausdorff distance between $K(\phi)$ and
$K(\psi)$.
\end{proof}

The \emph{power diagram} of a weighted point set $(Q,w)$ is a
decomposition of the space into convex polyhedra, one per point in
$Q$, defined by
$$\mathrm{Pow}_Q^w(q) = \{ x \in \Rsp^d;~ \forall p\in Q,~\nr{x - q}^2
+ w_q \leq \nr{x - p}^2 + w_p\}.$$ The following lemma shows the
relation between the vertices of the trace at infinity of $\phi_Q^w$
and unbounded cells in the power diagram of $(Q,w)$.
\begin{lemma}
Consider a finite weighted point set $(Q,w)$:
\begin{itemize}
\item[(i)] if the power cell of a point $q$ in $Q$ is unbounded, then $q$
  lies on the boundary of the polyhedron $K(\phi_Q^w)$;
\item[(i)] conversely, if $q$ is an extreme point in $K(\phi_Q^w)$, then
  the power cell of $q$ is unbounded.
\end{itemize}
\end{lemma}
Note that there might exist point $q$ that lie on the boundary of the
trace at infinity and whose power cell is empty. For instance,
consider $Q = \{ q_{-1},q_0,q_{1} \}$ in $\Rsp^2$ with $q_{i}=(0,i)$,
and weights $w_{-1} = w_1 = 0$ and $w_0 > 1$. Then $q_0$ lies on the
boundary of the convex hull while its power cell is empty.
\begin{proof}
Being convex, the power cell of a point $q$ is unbounded
if and only it contains a ray $\gamma_t := r + t u$, where $u$ is a
unit vector. By definition of the power cell, we have $\nr{q -
  \gamma_t}^2 + w_q \leq \nr{p-\gamma_t}^2 + w_p$ for all point $p$ in the set
$Q$. Expanding both sides expressions and simplifying, we get:
$$\nr{q-r}^2 - \nr{p-r}^2 + 2 t u \cdot (p-q) + w_q - w_p \leq 0$$ This
inequality holds for $t \to +\infty$, thus implying $u \cdot p \leq
u \cdot q$. Thus, $q$ lies on the boundary of $K(\phi)$, and $u$ is
an exterior normal vector to $K(\phi)$ at $q$.

Conversely, if $q$ is an extreme point of the convex polyhedron
$K(\phi)$, there must exist a unit vector $u$, such that $u \cdot
(p-q) \leq -\eps < 0$ for any point $p$ in $Q$ distinct from
$q$. Tracing back the above inequalities, and using the strict bound,
one can show that for any point $r$, the ray $\gamma_t := r + tu$
belongs to the power cell of $q$ for $t$ large enough.
\end{proof}
\subsection{Trace at infinity of the $k$-distance.} 
\label{sec:kdist}
Given a set of points $P$ in $\Rsp^d$ and an integer $k$ between one
and $\abs{P}$, the \emph{$k$-distance} to $P$ is defined by equation
\eqref{eq:dPk}.  The fact that this function is distance-like can be
seen in the following equivalent formulation, whose proof can be found
for instance in Proposition~3.1 of \cite{witnessed}:
\begin{equation}
\dd_{P,k}(x) = \left(\min_{\bar{p}} 
\sum_{i=1}^k \nr{x-\bar{p}}^2 + w_{\bar{p}} \right)^{1/2},
\label{eq:dPkbis}
\end{equation}
where the minimum is taken over the centroids $\bar{p}$ of $k$
distinct points in $P$, i.e. $\bar{p} = \frac{1}{k} \sum_{1\leq i\leq k} p_i$ and
where the weight is given by $w_{\bar{p}} := \frac{1}{k} \sum_{1\leq i\leq k}
\nr{\bar{p} - p_i}^2$.

\begin{definition}
We denote $K_{k}^d(P)$ the convex hull of the set of centroids of $k$
distinct points in $P$.  The support function of this polyhedron is
given by the following formula:
\begin{equation} \hh(K_k^d(P), u) := \max_{p_1,\hdots,p_k \in P} \frac{1}{k}
\sum_{i=1}^k u \cdot p_i
\label{eq:kpoly}
\end{equation}
\end{definition}

This polyhedron has been introduced first in
\cite{edelsbrunner1997cutting} under the name of \emph{$k$-set
  polyhedron} of the point set
$P$. Equations~\eqref{eq:dPkbis}--\eqref{eq:kpoly} above imply that
this polyhedron is the trace at infinity of the $k$-distance to
$P$. Moreover, the number of infinite order-$k$ Voronoi of $P$ is at
least equal to the number of extreme points in $K_k^d(P)$.

\tco{\smallskip}

\paragraph{Halving distance} 
When the number of points in $P$ is even and $k$ is equal to half this
number, we rename the $k$-distance the \emph{halving
  distance}. Similarly, we will refer to the order-$k$ Voronoi diagram
as the \emph{halving Voronoi diagram} and to the $k$-set polyhedron as
the \emph{halving polyhedron}.

\section{Approximating the sphere by \tco{\\}halving polyhedra}
In this section, we consider a family of random
polyhedra constructed as halving polyhedron of symmetric point sets on
the unit sphere. More precisely, we define:

\begin{definition} 
Given a set $P$  of $N$ points on the unit $(d-1)$-sphere, we
define $L_{N}^{d}(P)$ as the halving polyhedron of the symmetrization
of $P$, i.e. 
$$L_{N}^{d}(P) := K_N^d(P\cup(-P))$$ We obtain a (random) convex
polyhedron $L_{N}^{d} := L_{N}^d(P)$, by letting $P$ be an
independent random sampling of $N$ points on the unit $(d-1)$-sphere.
\end{definition}

Our main theorem consists in a lower bound on the probability of the
halving polyhedron $L_{N}^d$ to be Hausdorff-close to a ball centered
at the origin.

\begin{theorem} There exists an absolute constants $c> 0$ such that for every
positive number $\eta$, the inequality
$$\dH\left(\frac{1}{m_{d}} L_{N}^d, \Ball(0,1)\right)
  \leq \eta$$
holds with probability at least $1- 2 \exp\left[c \cdot \left(d
    \log(1/\delta) - N \eta^2\right)\right],$
where $\delta = \min(\eta,1/\sqrt{d})$.
\label{th:main}
\end{theorem}

The proof of this theorem is postponed to Section~\ref{sec:proof}. As
a first corollary, one can show that the random polyhedron $L_{N}^d$
is approximately round with high probability as soon as the
cardinality of the set of points sampled on the sphere grows faster
than $d\log d$, where $d$ is the ambient dimension.

\begin{corollary}
\label{coro}
For any $\kappa> 0$, there is a constant $C_\kappa$ such that for
$\eta>0$, $d \geq 1/\eta^2$ and
$N \geq \frac{d}{\eta^2} 
(\log(d)+C_\kappa)$
the following inequality 
$$\dH\left(\frac{1}{m_{d}} L_N^d, \Ball(0,1)\right) \leq \eta$$
holds with probability at least $1 - \exp(- \kappa d)$.
\end{corollary}

\begin{proof} From Theorem~\ref{th:main} the probability bound in the statement holds if
$- \kappa d \geq c (d \log(1/\delta) - N \eta^2)$. Since $\delta$ is
  equal to $1/\sqrt{d}$, this is the case if
$ N \eta^2 \geq d \left(\log(d) + \frac{2 \kappa}{c}\right)$.
\end{proof}


\section{Application: approximation of \tco{\\} distance-to-measures}
\label{sec:dtm}

The results of the previous section can be used to obtain a
probabilistic statement on the complexity of the halving distance to a
random point set on a high-dimensional sphere. We call
\emph{$\eps$-complexity} of a distance-like function $\phi:
\Rsp^d\to\Rsp$ the minimum number of sites that one needs in order to
be able to construct a distance-like function $\psi$ such that
the infinity norm $\nr{\phi - \psi}_\infty$ is at most $\eps$, i.e.
$$\LebNum(\phi,\eps) := \min \left\{ \abs{Q};~ \nr{\phi - \phi_{Q}^w}_\infty\leq \eps,~
\phi_{Q}^w \hbox{ as in } \eqref{eq:dl}\right\}.$$ The following theorem
provides a probabilistic lower bound on the $\eps$-complexity of a
family of distance-like functions.

\begin{theorem}
\label{th:approx}
For any constant $\kappa>0$ there exists a constant $C(\kappa)$ such
that the following hold. Let $\eta>0$, $d\geq 1/\eta^2$, and $S$ be
the symmetrization of an random point cloud of cardinality $N =
\frac{d}{\eta^2}(\log(d) + C_\kappa)$ on the unit sphere. Then, the
inequality
$$\LebNum\left(\dd_{S,N},
  m_{d}\eta\right) \geq 2 \sqrt{d} \left(\frac{N}{64
    d(\log(d)+C_\kappa)}\right)^{\frac{d-1}{4}} $$ 
holds with probability at least $1-\exp(- \kappa d)$.
\end{theorem}

Taking $\kappa = 1$ and $d = 1/\eta^2$, this implies:
\begin{corollary}
There exists a sequence of point clouds $S_d$ of cardinality
$d^2 (\log(d)+C_1)$ on the sphere $\Sph^{d-1}$ such that
$$\LebNum\left(\d_{S_d, \frac{1}{2}\abs{S_d}}, \frac{\sqrt{2/\pi}}{d}\right) \geq
2 \sqrt{d} \left(\frac{d}{64}\right)^{\frac{d-1}{4}}.$$ 
\end{corollary}



\begin{proof}[\Proof of Theorem~\ref{th:approx}]
The halving polyhedron $L_N^d(S)$ is by definition equal to the trace
at infinity of the halving distance $K(\dd_{S,N})$. Therefore, we can
apply Corollary~\ref{coro}: there exist a constant $C_\kappa$ such
that for a random point set $S$ distributed as in the statement of the
theorem, the following inequality holds:
\begin{equation}
\dH\left(\frac{1}{m_{d}} K(\dd_{S,N}),\B(0,1)\right) \leq
\eta.\label{eq:lbdh}
\end{equation}
with probability at least $1-\exp(- \kappa d)$.  We now consider a
deterministic point set $S$ on the sphere that satisfies the above
inequality, and we consider a distance-like function $\psi: x\mapsto
\min_{q\in Q} \nr{x-q}^2 + w_q$ such that the approximation error
$\Vert\psi - \dd_{S,N}\Vert_\infty$ is bounded by $m_{d} \eta$. By
definition of the complexity of $\dd_{S,N}$, our goal is to prove a
lower bound on the cardinality of the point set $Q$.  Using the
triangular inequality for the Hausdorff distance first, and then
Lemma~\ref{lemma:laguerre}.(ii), we get the following inequalities:
\begin{align*}
&\dH\left(\frac{1}{m_{d}} K(\psi), \B(0,1)\right)\\
&\qquad \leq \frac{1}{m_{d}} \dH(K(\psi), K(\d_{S,N})) 
+\dH\left(\frac{1}{m_{d}} K(\d_{S,N}), \B(0,1)\right) \\
&\qquad \leq \frac{1}{m_{d}} \Vert{\psi - \dd_{S,N}}\Vert_\infty + \eta 
 \leq 2\eta
\end{align*}
Now, recall that $K(\psi)$ is equal to the convex hull of the point
set $Q$. If we define $R$ as the rescaled set $\{ m_{d}^{-1} q;~q\in
Q\}$, the above inequality reads $\dH(\conv(R), \B(0,1)) \leq 2\eta.$

We have thus constructed a polyhedron that is within Hausdorff
distance $2\eta$ of the $d$-dimensional unit ball. The last remark
of~\cite{bronshteyn1975approximation} gives the following lower bound
on the number of vertices of such a polyhedron:
$$ \abs{Q} = \abs{R} \geq 2 \sqrt{d} (8\eta)^{- \frac{d-1}{2}}$$ To
conclude the proof, we simply replace $\eta$ by its expression in term
of the number of points $N$ and the dimension $d$, i.e. $\eta^2 :=
d(\log(d)+C_\kappa)/N$ to obtain
\begin{equation*}
\abs{Q}\geq 2 \sqrt{d} \left(\frac{N}{64 d(\log(d)+C_\kappa)}\right)^{\frac{d-1}{4}}\qedhere
\end{equation*}
\end{proof}

\section{Proof of the main theorem}
\label{sec:proof}

We start this section by showing a simple expression for the support
function of the halving polyhedra of a symmetric point set. For a
point set $P$, and $N := \abs{P}$, one has:
\begin{equation}
h(L_N^d(P), u) = \sum_{p\in P} \abs{p\cdot u}.
\label{lem:dist2}
\end{equation}
\begin{proof}[\Proof of Equation \eqref{lem:dist2}] Let $S = P \cup \{-P\}$ and
recall that by definition, the support function $\hh(L_{N}^d(P),u)$ is
equal to $\max \sum_{i=1}^N u \cdot p_i$, where the maximum is taken
on sets of $N$ distinct points in $S$. Choosing $\eps(p) = \pm 1$ such
that for every point in $p$ in $P$, $\eps(p) u \cdot p \geq 0$, one
easily sees that this maximum is attained for $\{p_1,\hdots,p_N\} =
\{\eps(p) p; p \in P\}$.
\end{proof}

This computation motivates the following definition.

\begin{definition}
We let $\LebSph^{d}_{1,u}$ be the measure on $[0,1]$ given by the
distribution of $\abs{u \cdot X}$ where $X$ is a uniformly distributed
random vector on the $(d-1)$-dimensional unit sphere, and $u$ is a
fixed unit vector. This measure turns out not to depend on $u$, so we
will denote it by $\LebSph^{d}_1$.
\end{definition}

The random polyhedron $L_{N}^d$ is defined as $L_{N}^d(P)$ where the
point set $P$ is obtained by drawing $N$ independent points on the
unit sphere. Equation~\eqref{lem:dist2} implies that for any unit
vector $u$ the distribution of values of $\hh(L_N^d,u)$ is given
by the formula:
\begin{equation}
 \frac{1}{N} \sum_{i=1}^N Y_i \label{eq:dist}, 
\end{equation}
where the $(Y_i)$ are $N$ independent random variables with
distribution $\lambda_1^{d}$.

\begin{lemma} The measure $\LebSph^{d}_1$
 has the following properties:
\begin{itemize} 
\item[(i)] $\LebSph^{d}_1$ is absolutely continuous with respect to
  the Lebesgue measure, with density
$$ f_{d}(t) := c_{d} (1-t^2)^{\frac{d-2}{2}}$$
the constant $c_{d}$ being chosen so that $f_{d}$ is the density of a probability measure, i.e. 
$c_{d} \int_{0}^1
  (1-t^2)^{\frac{d-2}{2}} \dd t = 1.$
\item[(ii)] the mean of $\LebSph^{d}_1$ is given by $m_{d} :=
  \frac{c_{d}}{d}$. Moreover, $m_{d}$ is equivalent to
  $\sqrt{2/(\pi d)}$ as $d\to\infty$.
\item[(iii)] The variance $\sigma^2_{d}$ of $\LebSph^{d}_1$ is
equivalent to $(1-2/\pi)/d$ as $d\to\infty$.
\end{itemize}
\end{lemma}

\begin{proof}
(i) Using Pythagoras theorem, one checks that the intersection of the
  hyperplane $\{ x \in \Rsp^d;~ u \cdot x = t \}$ with the unit sphere
  is a $(d-2)$-dimensional sphere with squared radius $(1-t^2)$. This
  implies the formula for $f_{d}$, for a certain constant
  $c_{d}$. To compute this constant one uses the fact that
  $\LebSph_1^{d}$ has unit mass, i.e. $$c_{d} \int_{0}^1
  (1-t^2)^{\frac{d-2}{2}} \dd t = 1.$$ Note that a formula of Wallis
  asserts that
$$ \lim_{d\to \infty} \sqrt{d} \int_0^1 (1-t^2)^{\frac{d}{2}} =
  \sqrt{\pi/2}.$$ This implies that $c_{d} \sim \sqrt{2 d/\pi}$.

\noindent
(ii) The function $t\mapsto t f_{d}(t)$ admits an explicit primitive:
$$g_{d}(t) := -\frac{c_{d}}{d} (1-t^2)^{\frac{d}{2}},$$ so that
  $m_{d} = g_{d}(1) - g_{d}(0)$ is as in the statement of the
  lemma. Thus, we get $m_{d} \simeq \sqrt{2/(\pi d)}$.

\noindent
 (iii) Integrating the following inequality between $0$ and $1$
$$t^2 (1-t^2)^{\frac{d-2}{2}} = (1-t^2)^{\frac{d-2}{2}} -
(1-t^2)(1-t^2)^{\frac{d-2}{2}},$$ one gets the following formula for
the second moment of $\LebSph^{d}_1$:
\begin{align*}
 c_{d} \int_0^1 t^2 (1-t^2)^{\frac{d-2}{2}} \dd t = 1 -
 \frac{c_{d}}{c_{d+2}}.
\end{align*}
Moreover, an integration by part gives
 $$ c_{d} \int_0^1 t^2 (1-t^2)^{\frac{d-2}{2}} \dd t =
\frac{c_{d}}{d} \int_0^1 (1-t^2)^{\frac{d}{2}} \dd t = \frac{1}{d}
\frac{c_{d}}{c_{d+2}}$$ These two equalities imply that
$c_{d}/c_{d+2} = d/(d+1)$, and using the formula for the mean given
above, we get:
\begin{equation*}
\sigma^2_{d} = \frac{1}{d} - \frac{c_{d}^2}{d^2} \sim \frac{1-2/\pi}{d}. \qedhere
\end{equation*}
\end{proof}

\begin{lemma} 
\label{lem:chernoff}
There exists a universal constant $c>0$ such that for any dimension
$d$, any $N>0$ and any set of directions $U$ in $\Sph^{d-1}$, one has
\begin{equation*}
\Prob\left(\max_{u\in U} \abs{\hh(L_{N}^d, u) - m_{d}} \geq \eta
m_{d}\right) \leq 2 \abs{U}
\exp\left(- c \cdot N \eta^2\right).
\end{equation*}
\end{lemma}


\begin{proof}
Consider $N$ random variables $Y_1,\hdots,Y_N$ with distribution
$\LebSph^{d}_1$. These random variable are bounded by $1$ and their
variance is $\sigma^2_{d}$. Applying Bernstein's inequality gives:
$$ \Prob\left(\abs{\frac{1}{N} \sum_{i=1}^N Y_i - m_{d}} \geq \eps\right) \leq
2 \exp\left(\frac{-N \eps^2}{2 \sigma^2_{d} + 2\eps/3}\right)$$
This implies that for a fixed direction $u$, 
$$
 \Prob\left(\abs{\hh(L_N^d, u) - m_{d}} \geq \eta m_{d}\right)
\leq 2 \exp\left(\frac{-N \eta^2 m_{d}^2}{2 \sigma^2_{d} + 2\eta
  m_{d}/3}\right)$$ Rescaling everything by a certain
constant $\kappa>0$, one gets
\begin{align*}
\tco{&}\Prob\left(\abs{\hh(L_N^d, u) - m_{d}} \geq \eta m_{d}\right) \tco{\\}
&\tco{\qquad}= \Prob\left(\abs{\hh(\kappa L_N^d, u) - \kappa m_{d}} \geq \kappa
\eta m_{d}\right) \\
&\tco{\qquad}\leq 2 \exp\left(\frac{-N \kappa^2 \eta^2 m_{d}^2}{2
  \kappa^2 \sigma^2_{d} + 2\eta \kappa m_{d}/3}\right)
\end{align*}
Letting $\kappa$ go to infinity, we obtain the following bound
$$
\Prob\left(\abs{\hh(L_N^d, u) - m_{d}} \geq \eta m_{d}\right)
\leq 2 \exp\left(-\frac{1}{2} N \eta^2 \frac{m_{d}^2}{\sigma_{d}^2}\right)
$$ This gives the desired estimate for a single direction using the
two estimates $m_{d} = \mathrm{O}(d^{-1/2})$ and $\sigma_{d}^2 =
\mathrm{O}(d^{-1})$. The conclusion of the lemma is obtained by a simple
application of the union bound.
\end{proof}


\begin{lemma}
If $K$ is contained in the ball $\B(0,r)$, the support function
$\hh(K,.)$ is $r$-Lipschitz.\label{lem:lipschitz}
\end{lemma}

\begin{proof}
Consider $u$ in the unit sphere, and $x$ in $K$ such that $\hh(K,u) =
u\cdot x$. For any vector $v$ in the unit sphere, 
\begin{align*}
\hh(K,v) = \max_{y\in K} v \cdot y 
      &\geq v \cdot x  = u \cdot x + (v-u) \cdot x\\
      &\geq \hh(K,u) - \nr{u - v} \nr{x} \\
&\geq \hh(K,u) - r \nr{u - v}.
\end{align*}
Swapping $u$ and $v$ gives the Lipschitz bound.
\end{proof}

A subset $U$ of the unit sphere $\Sph^{d-1}$ is called a
\emph{$\delta$-sample} or a $\delta$-covering if the union of the
Euclidean balls of radius $\delta$ centered at points of $U$ cover the
unit sphere.

\begin{lemma}
Consider a convex set $K$ contained in the unit ball $\Ball(0,1)$, and
two numbers $\lambda, \eta \in (0,1)$. Moreover, suppose
\begin{equation}
 \max_{u \in U} \abs{\hh(K,u) - \lambda} \leq \eta \lambda,
\label{eq:lemhyp}
\end{equation}
where $U$ is a $\delta$-sample of the unit sphere, with $\delta :=
\min(\lambda,\eta)$. Then, the Hausdorff distance between
$\frac{1}{\lambda} K$ and the unit ball is at most $5\eta$.
\label{lem:improv}
\end{lemma}

\begin{proof}
Note that, almost by definition, a convex set $K$ is included in the
ball $\Ball(0,r)$ if and only if its support function satisfies
$\nr{\hh(K,.)}_\infty \leq r$.  Assuming that the convex set $K$ is
contained in some ball $\Ball(0,r)$, we have:
\begin{align}
\tco{&}\nr{\hh(K,.) - \lambda}_\infty \tco{\notag}
\oco{&}= \max_{v \in \Sph^{d-1}} {\abs{\hh(K,v) -  \lambda}} \notag\\
&\tco{\qquad\qquad}\leq \max_{v \in \Sph^{d-1}}\left[\min_{u\in U} \abs{\hh(K,u) - \hh(K,v)} +
\abs{\hh(K,u) - \lambda}\right] \notag\\
&\tco{\qquad\qquad}\leq r \min(\lambda,\eta) + \eta \lambda.
\label{eq:ineqh}
\end{align}
The first inequality is obtained by applying the triangle inequality,
while the second one follows from the Lipschitz estimation of
Lemma~\ref{lem:lipschitz}, the fact that $U$ is a $\delta$-sample of
the unit sphere and from Equation~\ref{eq:lemhyp}.  Applying
Inequality~\eqref{eq:ineqh} with $r = 1$, $\eta < 1$ and using the
triangle inequality we get $\nr{\hh(K,.)-\lambda}_\infty \leq
2\lambda$. This implies that $K$ is contained in the sphere
$\Ball(0,3\lambda)$ and allows us to apply the same
inequality~\eqref{eq:ineqh} again with the smaller radius $r =
3\lambda$, implying
$$ \nr{\hh(K,.) - \lambda}_\infty \leq 3\lambda \eta + 2 \lambda\eta =
5\lambda \eta.$$ Dividing this last inequality by $\lambda$, and using
Equation~\eqref{eq:isom} implies the conclusion of the Lemma.
\end{proof}

\tco{\smallskip}

\paragraph{Proof of Theorem~\ref{th:main}}
Consider $\delta = \min(\eta,m_{d})$, and let $U$ be a
$\delta$-sample of the unit sphere with minimal cardinality.  The
cardinality of such a sample is bounded by $\abs{U} \leq
(\const/\delta)^{d}$, where the constant is absolute. Applying
Lemma~\ref{lem:improv} first and then Lemma~\ref{lem:chernoff} gives
us the following inequalities:
\begin{align*}
\tco{&} \Prob\left(\dH\left(\frac{1}{m_{d}} K_{N}^d,
\Ball(0,1)\right)\geq 5\eta\right) \oco{&}\tco{\\}
\tco{& \qquad\qquad}\leq \Prob\left(\max_{u\in U} \abs{\hh(L_{N}^d,
  u) - m_{d}} \geq \eta m_{d}\right) \\
&\tco{\qquad\qquad}\leq 2 \abs{U} \exp\left(- c \cdot N \eta^2\right).
\end{align*}
We conclude the proof by applying the upper bound on the cardinality of
the $\delta$-sample $U$ stated above, and by using the equivalent
$m_{d} \sim \sqrt{2/(\pi d)}$.


\section{Extension to other values of  $\frac{k}{N}$}
It is possible to extend some of the results to cases where ratio
$k/N$ is different from one half.  Let us consider the random polytope
$M^d_{N,k} = K_k^d(P)$, where $P$ is a set of $N$ random points
sampled uniformly and independently on the $(d-1)$-dimensional unit
sphere, and $k$ is between $1$ and $N$.  Lemma~\ref{lem:chernoff} can
be partially extended to this more general family of random polytopes.
Note however that the statement below does not provide any estimate
for the radius $r(d,N,k)$ as the number of points $N$ grows to
infinity with $k/N$ remaining constant. In particular, it is not
precise enough to generalize the lower bounds of
Theorem~\ref{th:main}.

\begin{lemma}
For any dimension $d$, and any numbers $N$ and $k$, there exist a
value $r := r(d,N,k)$ such that for any set of directions $U$ in
$\Sph^{d-1}$, and $N>0$ one has
\begin{equation*}
\Prob\left(\max_{u\in U} \abs{\hh(M_{N,k}^d, u) - r} \geq \eta r\right)
\leq 2 \abs{U} \exp\left(- \frac{1}{4}\frac{k^2}{N} \eta^2 r^2\right).
\end{equation*}
\label{lemma:chernbis}
\end{lemma}
\tco{\vspace{-.5cm}}

This lemma follows from a version of Chernoff's inequality adapted to
Lipschitz functions. Consider a function $F$ from the cube $[-1,1]^N$
to $\Rsp$, which is $\alpha$-Lipschitz with respect to the $\ell^1$
norm on the cube. Then, for any family $Y_1,\hdots,Y_N$ of
i.i.d. random variables taking values in $[-1,1]$ one has:
\begin{equation}
\Prob(\abs{F(Y_1,\hdots,Y_N) - \Exp F}\geq \eps) \leq
 2 \exp\left(\frac{-\eps^2}{4 \alpha^2 N}\right).
\label{eq:chernofflip}
\end{equation}

\begin{proof}[\Proof of Lemma~\ref{lemma:chernbis}]
We consider the map $F_N^k$ from the cube $[-1,1]^N$ to $\Rsp$ defined
by:
$$F_N^k(x) := \max \left\{ \frac{1}{k} \sum_{j=1}^k x_{i_j};~ 1\leq
i_1< \hdots<i_k \leq N \right\},$$ where $x_i$ denotes the $i$th
coordinate of $x$. We also consider the measure $\mu$ obtained by
pushing forward the $(d-1)$-area measure on the unit sphere of
$\Rsp^d$ by the map $x \mapsto u\cdot x$, for some direction $u$ in
the unit sphere. As in the halving case (cf \eqref{eq:dist}), the
distribution of $\hh(M_{N,k}^d,u)$ is given by the distribution of
$F_N^k(Y_1,\hdots, Y_N)$, where $Y_1,\hdots Y_N$ are i.i.d random
variables with distribution $\mu$.

Moreover, the map $F_N^k$ is $\frac{1}{k}$-Lipschitz with respect to
the $\ell^1$ norm on the cube $[-1,1]^N$. Indeed, given two points
$x,y$ in $[-1,1]^N$ and a set of indices $i_1<\hdots<i_k$ 
corresponding to the maximum in the definition of $F_N^k(x)$, one has
\begin{align*}
k F_N^k(x) = x_{i_1} + \hdots + x_{i_k} &\leq y_{i_1} + \hdots + y_{i_k}
+ \nr{x-y}_{\ell^1}\\ &\leq k F_N^k(y) + \nr{x-y}_{\ell^1}
\end{align*}
Thus, we can apply Chernoff's inequality \eqref{eq:chernofflip}:
\begin{equation*}
\tco{\hspace{-.1cm}}\Prob\left(\abs{F_N^k(Y_1,\hdots,Y_N) - \Exp F_N^k} \geq \eps\right) \leq
2 \exp\left( - \frac{\eps^2 k^2}{4 N}\right)
\tco{\hspace{-.1cm}}\label{eq:chernoff}
\end{equation*} where $m = k/N$. The Lemma follows by setting
$r := \Exp F_N^k$, $\eps = \eta r$ in the equation and by using the
union bound.
\end{proof}

\smallskip
\paragraph{Acknowledgements} This work was partially funded by French ANR grant
GIGA ANR-09-BLAN-0331-01.

\oco{\bibliographystyle{amsalpha}}

\providecommand{\bysame}{\leavevmode\hbox to3em{\hrulefill}\thinspace}
\providecommand{\MR}{\relax\ifhmode\unskip\space\fi MR }
\providecommand{\MRhref}[2]{%
  \href{http://www.ams.org/mathscinet-getitem?mr=#1}{#2}
}
\providecommand{\href}[2]{#2}

\end{document}